\documentclass[12pt]{article}
\usepackage{makeidx}
\usepackage{amsfonts}
\usepackage{amsmath}
\usepackage{makeidx}
\usepackage{amssymb}
\usepackage{color}
\usepackage{enumerate}
\usepackage{epsfig}
\usepackage[utf8]{inputenc}
\usepackage[T1]{fontenc}
\usepackage{inputenc}
\usepackage{slashbox}

\newtheorem{theorem}{Theorem}

\newtheorem{aexample}[theorem]{Example}

\newtheorem{lemma}[theorem]{Lemma}

\newtheorem{aproblem}[theorem]{Problem}

\newtheorem{acomment}[theorem]{Comment}

\newtheorem{aremark}[theorem]{Remark}
\newenvironment{remark}{\begin{aremark}\rm}{\end{aremark}}
\newtheorem{result}[theorem]{Result}
\newenvironment{proof}[1][Proof]{\noindent\textbf{#1.} }{\ \rule{0.5em}{0.5em}}
\numberwithin{equation}{section} \numberwithin{theorem}{section}
\begin{document}

\title{The absence of the selfaveraging property \\
of the entanglement entropy \\
of disordered free fermions \\
in one dimension}
\author{L. Pastur and V. Slavin \\
%EndAName
B. Verkin Institute for Low Temperatures Physics \\
and Engineering, Kharkiv, Ukraine}
\date{}
\maketitle

\begin{abstract}
We consider the macroscopic system of free lattice fermions in one
dimension assuming that the one-body Hamiltonian of the system is the one
dimensional discrete Schr\"odinger operator with independent identically
distributed random potential. We show analytically and numerically that the
variance of the entanglement entropy of the segment $[-M,M]$ of the system
is bounded away from zero as $M\rightarrow \infty $. This manifests the
absence of the selfaveraging property of the entanglement entropy in our
model, meaning that in the one-dimensional case the complete description of
the entanglement entropy is provided by its whole probability distribution.
This also may be contrasted the case of dimension two or more, where the
variance of the entanglement entropy per unit surface area vanishes as $%
M\rightarrow \infty $ \cite{El-Co:17}, thereby guaranteeing the
representativity of its mean for large $M$ in the multidimensional case.
%\keywords{entanglement entropy \and free fermions \and Anderson localization}

PACS numbers 03.67.Mn, 03.67, 05.30.Fk, 72.15.Rn
\end{abstract}

%\date{}

\section{Introduction}

This note is an addition to the paper \cite{El-Co:17} by A. Elgart, M.
Shcherbina and the first author in which it is proved the following.
Consider the macroscopic system of free disordered fermions living on the $d$%
-dimensional lattice $\mathbb{Z}^{d}$ and having the discrete Schr\"{o}dinger operator
\begin{equation}
H=-\Delta _{d}+V,  \label{SD}
\end{equation}%
as the one-body Hamiltonian. Here%
\begin{equation}
(\Delta _{d}u)(x)=-\sum_{|x-y|=1}u(y)+2du(x),\;x\in \mathbb{Z}^{d}
\label{lad}
\end{equation}%
is the $d$-dimensional discrete Laplacian and
\begin{equation}
V=\{V(x)\}_{x\in \mathbb{Z}^{d}}  \label{pot}
\end{equation}%
is the random ergodic potential. Assume that the Fermi energy $E$ of the
system lies in the exponentially localized part of spectrum of $H$. This
means that the Fermi projection
\begin{equation}
P=\{P(x,y)\}_{x,y\in \mathbb{Z}^{d}}=\{\mathcal{E}_{H}(x,y;(-\infty
,E))\}_{x,y\in \mathbb{Z}^{d}}  \label{fp}
\end{equation}%
of $H$, i.e., its spectral projection measure $\mathcal{E}_{H}$
corresponding to the spectral interval \ $(-\infty ,E)$, admits the bound
\begin{equation}
\mathbf{E}\{|P(x,y)|\}\leq Ce^{-\gamma |x-y|},\;x,y\in \mathbb{Z}^{d}
\label{b_P}
\end{equation}%
for some $C<\infty $ and $\gamma >0$. Here and below the symbol $\mathbf{E}%
\{...\}$ denotes the expectation with respect to the random potential.

We refer the reader to \cite{El-Co:17} for the discussion of the cases where
the bound (\ref{b_P}) holds and guaranties the pure point spectrum of (\ref%
{SD}) with exponentially decaying eigenfunctions (exponential localization).
It is important for us in this paper that in the one-dimensional case $d=1$
the bound holds on the whole spectrum of $H$ if the potential (\ref{pot}) is
the collection of independent identically distributed (i.i.d.) random
variables.

Given the lattice cube (the block)
\begin{equation*}
\Lambda =[-M,M]^{d}\subset \mathbb{Z}^{d}
\end{equation*}%
of a macroscopic quantum system, we define the entanglement (von Neumann)
entropy of the block as
\begin{equation}
S_{\Lambda }=\mathrm{Tr}\ \rho_\Lambda \log \rho_\Lambda,  \label{vN}
\end{equation}%
where $\rho_\Lambda$ is the reduced density matrix for the block. In
particular, for the macroscopic system of free lattice spinless fermions at
zero temperature (i.e., at their ground state) we have \cite%
{Ab-St:15,Ar-Co:14}%
\begin{equation}
S_{\Lambda }=\mathrm{Tr}\ h(P_{\Lambda }),  \label{ee}
\end{equation}%
where%
\begin{equation}
h(x)=-x\log x-(1-x)\log (1-x),\;x\in \lbrack 0,1]  \label{h}
\end{equation}%
and
\begin{equation}
P_{\Lambda }=\{P(x,y)\}_{x,y\in \Lambda }  \label{fpl}
\end{equation}%
is the restriction of the Fermi projection \ref{fp}) to the block $\Lambda $.

It is proved in \cite{El-Co:17} that for any ergodic potential satisfying
condition (\ref{b_P}) of the exponential localization the entanglement
entropy satisfies the area law in the mean, i.e., there exists the limit%
\begin{equation}
\lim_{L\rightarrow \infty }\mathbf{E}\{L^{(d-1)}S_{\Lambda }\}=2d\,\mathbf{E}%
\{\mathrm{Tr}\,h(P_{\mathbb{Z}_{+}^{d}})\},\;L=2M+1,  \label{alm}
\end{equation}%
where $P_{\mathbb{Z}_{+}^{d}}=\{P(x,y)\}_{x,y\in \mathbb{Z}_{+}^{d}}$ is the
restriction of the Fermi projection (\ref{fp}) to the $d$-dimensional
lattice half-space%
\begin{equation*}
\mathbb{Z}_{+}^{d}=\mathbb{Z}_{+}\times \mathbb{Z}^{d-1},\;\mathbb{Z}%
_{+}=[0,1,..).
\end{equation*}%
See \cite{Ca-Co:11,Ei-Co:11,La:15,Le-Co:13} for various results on the
validity of the area law and its violation in translation invariant
(non-random) systems.

It was also shown in \cite{El-Co:17} that if the random potential is a
collection of i.i.d. random variables and (\ref{b_P}) holds, then there
exist some $C_{d}<\infty $ and $a_{d}>0$ such that
\begin{align}\label{vare}
\mathbf{Var}\{L^{(d-1)}S_{\Lambda }\}:= &\mathbf{E}\{(L^{(d-1)}S_{\Lambda
})^{2}\}-\mathbf{E}^{2}\{L^{(d-1)}S_{\Lambda }\}
\leq C_{d}/L^{a_{d}},\;\;d\geq 2,
\end{align}%
i.e., that the fluctuations of the entanglement entropy per unit surface
area vanish as $L\rightarrow \infty $.

The relations (\ref{alm}) and (\ref{vare}) imply that if $d \ge 2$ then
the entanglement entropy per unit surface area possesses the
selfaveraging property (see, e.g., \cite{Da-Co:14,Di-Co:98,LGP,Pa-Fi:92} for
discussion and use of the property in the condensed matter theory, spectral
theory and the quantum information theory where it is known as the
typicality).

On the other hand, it follows from the numerical results of \cite{Pa-Sl:14}
that for $d=1$ the fluctuations of the entanglement entropy of the lattice
segment $\Lambda =[-M,M]$ do not vanish as $M\rightarrow \dot{\infty}$ and
according to \cite{El-Co:17} we have for every typical realization (with
probability 1)%
\begin{equation}
S_{[-M,M]} =S_{\mathbb{Z}}+o(1),\;M\rightarrow \infty ,  \label{sp1} \\
\end{equation}%
where
\begin{equation}\label{sz}
S_{\mathbb{Z}} =S_{+}(T^{M}\omega )+S_{-}(T^{-M}\omega )
\end{equation}%
and
\begin{equation}
S_{\pm }(\omega )=\mathrm{Tr}\,h(P_{\mathbb{Z}_{\mp }}(\omega )),\ .
\label{spm}
\end{equation}
are non-zero with probability 1. Here and below $\omega =\{V(x)\}_{x\in
\mathbb{Z}}$ denotes a realization of random ergodic potential and $T$ is
the shift operator acting in the space of realizations of potential as $%
T\omega =\{V(x+1)\}_{x\in \mathbb{Z}}$. This suggests that for
i.i.d. random potential the entanglement entropy (\ref{vN}) (see also (\ref{sp1}) -- (\ref{spm})) of \ disordered free fermions is
not a selfaveraging quantity for $d=1$ .

In this note we confirm the suggestion by establishing an $M$-independent
and strictly positive lower bound on the variance of the entanglement
entropy for $d=1$. Unfortunately, the class of random i.i.d. potentials, for
which this results is established, is somewhat limited (see, e.g. Remark \ref{r:f}). However, since the absence of selfaveraging property is not
completely common and sufficiently studied in the theory of disordered
systems, we believe that our result is of certain interest.

The paper is organized as follows. In Section 2 we formulate and prove
analytically a strictly positive for all sufficiently large $L=2M+1$ lower bound
for the entanglement entropy (\ref{sp1}) -- (\ref{spm}) of one dimensional free fermions. In Section 3
we present our numerical results which confirm and illustrate the analytic
results. Section 4 contains several auxiliary facts which we need in Section
2.

\section{Analytical results.}
Here we formulate and prove a lower bound for the entanglement entropy
(\ref{sp1}) -- (\ref{spm}) of disordered free fermions modulo several
technical facts of Section 4.

\begin{result}
Consider the one dimensional macroscopic system of free lattice fermions
whose one-body Hamiltonian is the discrete Schr\"{o}dinger operator
(\ref{SD}) with i.i.d. random potential (\ref{pot}). Assume that the
common probability distribution of $V(x),\;x\in \mathbb{Z}$ has a
bounded density $f$ such that

\medskip (i) $\mathrm{supp\ }f=[0,\infty )$ and for some $\kappa >0$
\begin{equation}
\int_{0}^{\infty }v^{\kappa }f(v)dv<\infty ;  \label{cl2}
\end{equation}

\medskip (ii) the quantity
\begin{equation}
F(t)=J(t)-1,\;J(t)=\int_{0}^{\infty }\frac{f^{2}(v-t)}{f(v)}dv  \label{F}
\end{equation}%
is finite for all sufficiently large $t>0.$

Then there exist a sufficiently large $t_{0}>0$ and $M_{0}>0$ such that we
have for the entanglement entropy (\ref{sp1}) -- (\ref{spm}) uniformly in $M>M_{0}$
\begin{equation}
\mathbf{Var}\{S_{[-M,M]}\}:=\mathbf{E}\{S_{[-M,M]}^{2}\}-\mathbf{E}%
^{2}\{S_{[-M,M]}\}\geq A>0,  \label{var}
\end{equation}

\begin{equation}
A=\mathbf{E}^{2}\{S_{-}\}/F(t_{0})  \label{AF}
\end{equation}%
and $S_{-}$ is defined in (\ref{spm}).
\end{result}

\begin{remark}
\label{r:f} (i) It is easy to show (see (\ref{Fpos}) below) that $F\geq 0$.
Moreover, $F$ is unbounded as $t\rightarrow \infty $. Indeed, we have from (%
\ref{F}) and the Jensen inequality%
\begin{eqnarray*}
J(t) &=&\int_{0}^{\infty }\frac{f(v)}{f(v+t)}f(v)dv \\
&\geq &\left( \int_{0}^{\infty }\left( \frac{f(v)}{f(v+t)}\right)
^{-1}f(v)dv\right) ^{-1}=\left(\int_{t}^{\infty }f(v)dv\right)^{-1}\rightarrow \infty
,\;t\rightarrow \infty .
\end{eqnarray*}%
Thus, $A$ in (\ref{var}) -- (\ref{AF}) can be rather small (see, however,
Fig. \ref{fig5}). Note also that above
lower bound for $I$ is exact for the density $f(v)=\delta^{-1} e^{-v/\delta }$.

\medskip (ii) Condition (i) of the result can be replaced by that for the
support of $f$ to be bounded from below. However, a compact support is not
allowed, since in this case $J(t)$ in (\ref{F}) is not well defined for
large $t$, since the supports of the numerator and denominator in $J(t)$ do
not intersect. Moreover, even if the support of $f$ is the positive
semi-axis, $f$ should not have zeros of the order 1 and higher.

\medskip (iii) Our result is also valid for the quantum R\'enyi entropy
defined in general as (cf. (\ref{vN}))
\begin{equation}  \label{Rg}
S_\Lambda^{(\alpha)}= (1-\alpha)^{-1} \log \mathrm{Tr} \
\rho_\Lambda^\alpha, \; \alpha \in (0,\infty).
\end{equation}
The case $\alpha=1$ corresponds to the von Neumann entropy (\ref{vN}). For
the free lattice fermions in their ground state we have
\begin{equation}  \label{Rf}
S_\Lambda^{(\alpha)}= \mathrm{Tr} \ h_{\alpha}(P_\Lambda), \; \alpha \in
(0,\infty).
\end{equation}
where (cf. (\ref{h}))
\begin{equation}  \label{hal}
h_{\alpha}= (1-\alpha)^{-1} \log (x^\alpha + (1-x)^\alpha), \; x \in [0,1].
\end{equation}
In particular,
\begin{equation}  \label{hain}
h_{1}(x)=h(x), \;\; h_\infty(x)= - \log(\max\{1-x,x\}).
\end{equation}
For the details of the proof see Remark \ref{r:weyl} (ii) below.
\end{remark}

\textbf{Proof of result}. It follows from (\ref{alm}) for $d=1$ (or from (%
\ref{sp1}) -- (\ref{spm})) that
\begin{eqnarray}  \label{fim}
\mathbf{E}\{S_{[-M,M]}\} &=&\mathbf{E}\{S_{+}(T^{M}\omega
)+S_{-}(T^{-M}\omega )\}+o(1) \\
&=&2\mathbf{E}\{S_{-}(\omega )\}+o(1),\;M\rightarrow \infty ,  \notag
\end{eqnarray}%
and in obtaining the second equality we used the shift and the reflection
invariance of the probability distribution of the infinite sequence $%
V=\{V(x)\}_{x\in \mathbb{Z}}$ of i.i.d. random variables, see \cite{El-Co:17}.

Likewise, repeating almost literally the proof of (\ref{alm}) for $d=1$ in
\cite{El-Co:17}, we obtain%
\begin{eqnarray}
\mathbf{E}\{S_{[-M,M]}^{2}\} &=&\mathbf{E}\{(S_{+}(T^{M}\omega
)+S_{-}(T^{-M}\omega ))^{2}\}+o(1)  \label{es2} \\
&=&2\mathbf{E}\{(S_{-}(\omega ))^{2}\}+2\mathbf{E}\{S_{+}(T^{2M}\omega
)S_{-}(\omega )\}+o(1),\;M\rightarrow \infty .  \notag
\end{eqnarray}%
Since the infinite sequence $V=\{V(x)\}_{x\in \mathbb{Z}}$ of i.i.d. random
variables is a mixing stationary process (see e.g. \cite{Sh:95}, Section V.2
), we have%
\begin{eqnarray}  \label{ess}
\mathbf{E}\{S_{+}(T^{2M}\omega )S_{-}(\omega )\} &=&\mathbf{E}\{S_{+}(\omega
)\}\mathbf{E}\{S_{-}(\omega )\}+o(1)  \label{cs2} \\
&=&\mathbf{E}^{2}\{S_{-}(\omega )\}+o(1),\;M\rightarrow \infty .  \notag
\end{eqnarray}%
Combining (\ref{fim}) -- (\ref{ess}), we obtain%
\begin{equation}
\mathbf{Var}\{S_{[-M,M]}\}=2\mathbf{Var}\{S_{-}\}+o(1),\;M\rightarrow \infty
.  \label{v2v}
\end{equation}%
Hence, we have for the variance of the limiting entanglement entropy $S_{%
\mathbb{Z}}$ in \ (\ref{fim})%
\begin{equation}\label{szsm}
\mathbf{Var}\{S_{\mathbb{Z}}\}=2\mathbf{Var}\{S_{-}\}
\end{equation}
and it suffices to show that $\mathbf{Var}\{S_{-}\}$ is strictly positive.

To this end we start with the inequality
\begin{equation*}
\mathbf{Var}\{\varphi (\xi ,\eta )\}\geq \mathbf{Var}\{\mathbf{E}\{\varphi
(\xi ,\eta )|\xi \}\}
\end{equation*}%
involving the conditional expectation and valid for any random
(multi-component in general) variables $\xi $ and $\eta $ and a function $%
\varphi $. Choosing here $\xi =V(0)$, $\eta =\{V(x)\}_{x\neq 0}$ and $%
\varphi =S_{-}$, we obtain%
\begin{equation*}
\mathbf{Var}\{S_{-}\}\geq \mathbf{Var}\{\mathbf{E}\{S_{-}|V(0)\}\}.
\end{equation*}%
Next, we will use Lemma \ref{l:cr} with $\xi =V(0)$ and $\varphi (\xi )=$ $%
\mathbf{E}\{S_{-}|V(0)=\xi \}$ yielding%
\begin{equation}
\mathbf{Var}\{S_{-}\}\geq \left. (\mathbf{E}\{S_{-}\}-\mathbf{E}%
\{S_{-}^{t}\})^{2}\right/ F(t),  \label{varl}
\end{equation}%
where $S_{-}^{t}$ is the entanglement entropy (\ref{ee}) -- (\ref{fpl})
corresponding to the Schr\"{o}dinger operator $H^{t}$ (see (\ref{SD}) -- (%
\ref{pot})) in which the potential $V(0)$ at the origin is replaced by $%
V(0)+t, \; t>0$%
\begin{equation}
H^{t}=H|_{V(0)\rightarrow V(0)+t}.  \label{ht}
\end{equation}%
Combining \ (\ref{v2v}) -- (\ref{varl}), we obtain for any $t>0$%
\begin{equation}
\mathbf{Var}\{S_{\mathbb{Z}}\}:=\lim_{M\rightarrow \infty }\mathbf{Var}\{S_{[-M,M]}\}\geq 2(\mathbf{E}%
\{S_{-}\}-\mathbf{E}\{S_{-}^{t}\})^{2}/F(t),  \label{vart}
\end{equation}%
where $F(t)$ is defined in (\ref{F}).

\bigskip

We will prove below that%
\begin{equation}
\lim_{t\rightarrow \infty }\mathbf{E}\{S_{-}^{t}\}=0.  \label{leg}
\end{equation}%
Thus, there exists $\varepsilon (t)\rightarrow 0,\;t\rightarrow \infty $
(see, e.g (\ref{rem})) and sufficiently large $t_{0}$ such that we have in
view of (\ref{spm})%
\begin{equation}
\mathbf{E}\{S_{-}\}-\mathbf{E}\{S_{-}^{t}\})=(1-\varepsilon (t))\mathbf{E}%
\{S_{-}\},\;\;  \label{eeps}
\end{equation}%
and then (\ref{v2v}) -- (\ref{varl}) yield (\ref{var}) -- (\ref{AF}) upon
choosing sufficiently large $M_{0}$ and $t_{0}$ and assuming that $M\geq
M_{0}$ and $t\geq t_{0}$ to provide sufficiently small error terms in (\ref%
{fim}) -- (\ref{v2v}) and (\ref{eeps}).

Let us prove (\ref{leg}). Since the potential is a collection of i.i.d.
random variables satisfying condition (\ref{cl2}), the spectrum of $H$ is
the positive semi-axis (see Corollary 4.23 in \cite{Pa-Fi:92}). The same is
true for the spectrum of $H^{t}$ of (\ref{ht}) since $t>0$. Hence, we have in view of (\ref{fp})
\begin{equation}
P=\mathcal{E}_{H}((0,E)),\quad P^{t}=\mathcal{E}_{H^{t}}((0,E)),\quad E>0.
\label{PE}
\end{equation}%
It follows  from the proof of Lemma 4.5 of \cite{El-Co:17} that we have for some $t$%
-independent $C_0 < \infty$ and any $\alpha \in (0,1)$:%
\begin{align}
\mathbf{E}\{S_{-}^t\})& \leq C_{0}\sum_{x=0}^{\infty }\left( \sum_{y=-\infty
}^{-1}\mathbf{E}\{|P^{t}(x,y)|^{2}\}\right) ^{\alpha }  \label{espt} \\
& \leq C_{0}\mathbf{E}^{\alpha }\{P^{t}(0,0)\}+C_{0}\sum_{x=1}^{\infty
}\left( \sum_{y=-\infty }^{-1}\mathbf{E}\{|P^{t}(x,y)|^{2}\}\right) ^{\alpha
},  \notag
\end{align}%
where we took into account the inequality $(a+b)^{\alpha }\leq a^{\alpha
}+b^{\alpha },\;\alpha \in (0,1)$ and that $P^{t}$ is an orthogonal
projection, hence
\begin{equation*}
\sum_{y=-\infty }^{-1}|P^{t}(0,y)|^{2}\leq \sum_{y=-\infty }^{\infty
}|P^{t}(0,y)|^{2        }=P^{t}(0,0).
\end{equation*}%
Now, (\ref{espt}) and Lemma \ref{l:pvan} below yield
\begin{equation}  \label{rem}
\mathbf{E}\{S^t_{-}\} \leq C_{\alpha,s}/ (t-E)^{\alpha s},
\end{equation}
where $C_{\alpha, s}$ does not depend on $t$. This implies (\ref{leg}) $%
\blacksquare$.

\section{Numerical results}

Recall that we consider the free disordered fermions whose one-body Hamiltonian is
the Schr\"odinger operator (\ref{SD}) -- (\ref{pot}) with random
i.i.d. potential for $d=1$. We will use three particular random potentials with the
following on-site probability density $f$ where $\delta $ is the disorder
parameter which will be varied below:

(i) uniform%
\begin{equation}
f(v)=\frac{1}{\delta }\left\{
\begin{array}{cc}
1, & v\in \lbrack 0,\delta ], \\
0, & v\notin \lbrack 0,\delta ];%
\end{array}%
\right.  \label{un}
\end{equation}

(ii) exponential%
\begin{equation}
f(v)=\frac{1}{\delta }\left\{
\begin{array}{cc}
e^{-v/\delta }, & v\geq 0, \\
0, & v<0;%
\end{array}%
\right.  \label{ex}
\end{equation}

(iii) "half"-Cauchy%
\begin{equation}
f(v)=\frac{2\delta }{\pi }\left\{
\begin{array}{cc}
(v^{2}+\delta^2 )^{-1}, & v\geq 0, \\
0, & v<0.%
\end{array}
\right.
\label{ca}
\end{equation}

\begin{figure}[ht]
\begin{minipage}[ht]{0.5\linewidth}
\center{\includegraphics[width=0.99\linewidth]{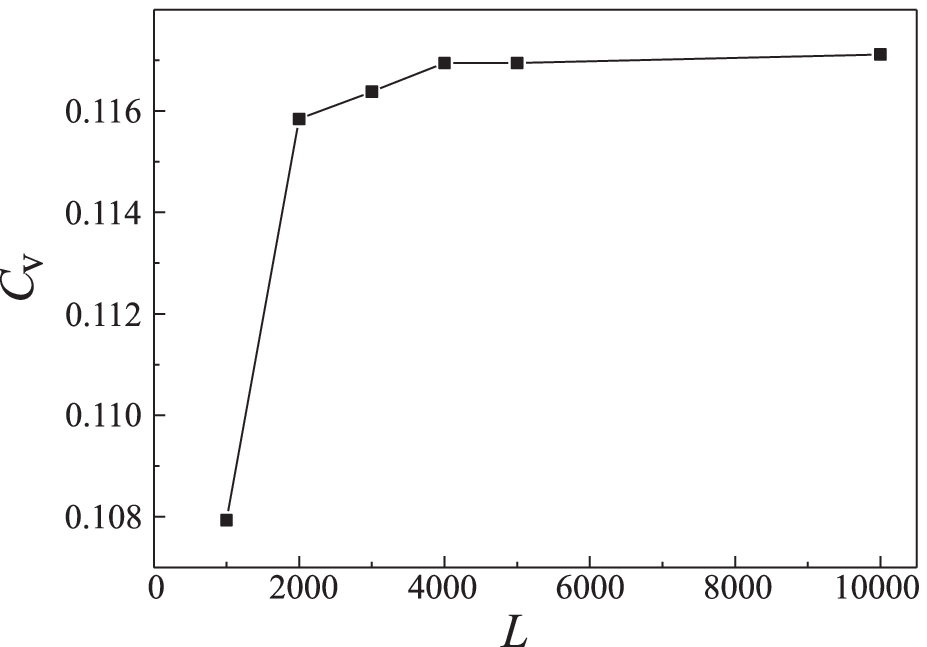} \\a)}
\end{minipage}
\hfill
\begin{minipage}[ht]{0.5\linewidth}
\center{\includegraphics[width=0.99\linewidth]{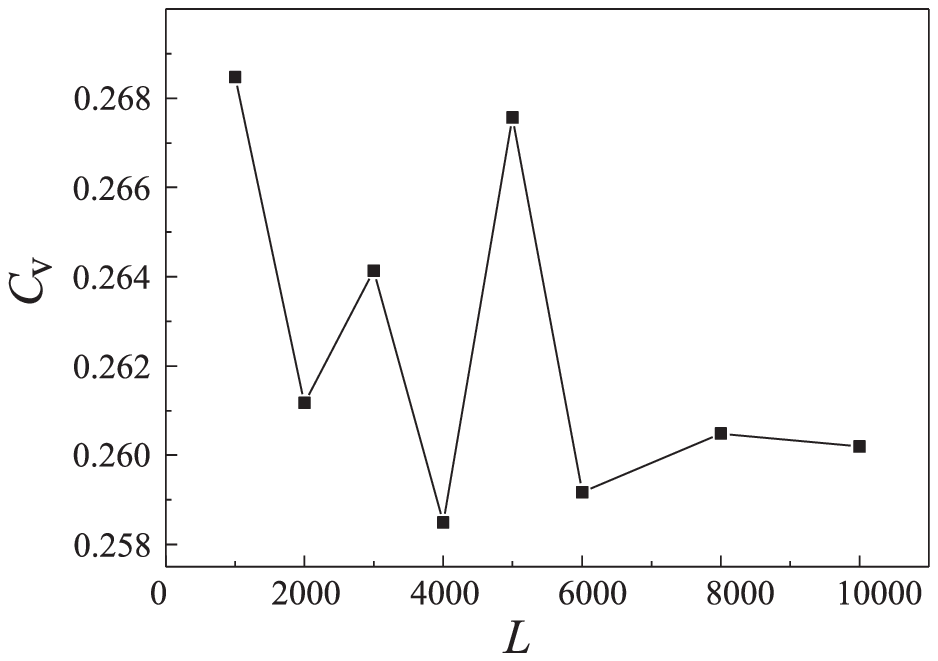} \\b)}
\end{minipage}
\caption{The coefficient of variation (relative standard deviation)
$\mathbf{C_V} \{S_{[-M,M]}\}$, see (\ref{cv}), as a function of  the block length $L=2M+1$
for: a) uniform distribution (\ref{un}), $\delta=1$; b) exponential distribution (\ref{ex}), $\delta=1$.}
\label{fig11}
\end{figure}

Note that the uniform distribution does not satisfies condition (i) of our
analytical result(\ref{var}) -- (\ref{AF}).

Recall also that we follow in this paper the widely
accepted asymptotic regime of the bipartite setting of quantum information
theory according to which
the size $N$ of the whole macroscopic system and the size $L=2M+1$ of the block $\Lambda=[-M,M]$
are related as
\begin{equation}
1<<L<<N.  \label{bitp}
\end{equation}

A number of formulas of Sections 1 and 2 are written for infinite
system, i.e., for the implementation $N=\infty$ of the r.h.s. inequality
(\ref{bitp}), see, e.g. (\ref{sz}), (\ref{spm}) (\ref{var}), (\ref{varl}) and (\ref{vart}). In our numerical results below we deal with finite systems whose size (length)
is $N=$10000 (Fig. 1) and $N=$5000 (Fig. 2 -- Fig. 5). The block size
(length) $L$  is varied from 1 to 10000 in Fig. 1 and is $L=$%
2500 in Fig. 2 -- Fig. 5. The operation of the expectation with respect to
the realizations of the corresponding random potential (see (\ref{un}) -- (\ref{ca}))
is carried out via the arithmetic mean of numerical results obtained for 2000
realizations of the potential.

A natural length scale in problems involving the eigenstate localization due to the
disordered medium is the localization radius of eigenstates at the Fermi energy.
The radius (the inverse Lyapunov exponent) is, roughly, the inverse rate
$\gamma ^{-1}$ of decay of the bound (\ref{b_P}) for the Fermi projection
(\ref{fp}). It follows from Table. \ref{ta1} that the
localization radius at the Fermi energy $E=1$ and the values of the disorder
parameter $\delta$ of (\ref{un}) -- (\ref{ca}), which we use in our computations, is
small enough to guarantee sufficiently strong inequalities in (\ref{bitp})
except, possibly, the case of weak ($\delta =0.2$) disorder for the uniform
distribution (\ref{un}).

\begin{table}
\centering
\begin{tabular}{ |l|c|c|c|c|c| }\hline
\backslashbox{Distribution}{$\delta$}
            & 0.2  & 0.4 & 0.6 & 0.8 & 1.0 \\ \hline
Uniform     & 1650 & 360 & 155 & 75  & 41  \\ \hline
Exponential & 134  & 30  & 13  & 7   & 5   \\ \hline
Half-Cauchy & 72   & 17  & 7   & 4   & 3   \\ \hline
\end{tabular}
\caption{The localization radius at the Fermi energy $E=1$
for three random potentials given by the uniform (\ref{un}),
exponential (\ref{ex}) and "half"-Cauchy (\ref{ca}) distributions
and different values of disorder parameter $\delta$.}
\label{ta1}
\end{table}

Fig. 1 depicts the finite system ($N=10000$) versions of the coefficient of variation (relative standard deviation)
\begin{equation}\label{cv}
\mathbf{C_V}\{S_{[-M,M]}\}=(\mathbf{Var}\{S_{[-M,M]}\})^{1/2}/\mathbf{E}\{S_{[-M,M]}\}
\end{equation}
of the entanglement entropy (\ref{ee}) -- (\ref{fpl}) as a function of the
block length $L=2M+1$ for the uniform (\ref{un}) and the exponential (\ref{ex}) probability distributions of the on-site potential. Both plots exhibit the tendency to approach a non-zero value for large $L$. This shows, in particular, that the error terms in formulas (\ref{fim}), (\ref{es2}), (\ref{cs2}) and (\ref{v2v}) decay sufficiently fast as $L \to \infty$, although with certain finite size effects demonstrated in Fig. 1b).

Fig. 2a) contains the plots of the probability density of the entanglement entropy $S_{[-M,M]}$ for $N=5000$ and $L=2M+1=2500$ and the uniform (\ref{un}) on-site probability distribution of the potential (\ref{pot}). We see that the corresponding plots move monotonically to the origin as disorder grows. This seems a natural property, since the entanglement entropy, being a measure of quantum coherence, should decrease with the increase of the disorder which inhibits
the coherence (see also formula (\ref{edis}) below).

Furthermore, the plots corresponding to the weak and medium disorder ($\delta=0.2,\ 0.4, \ 1.0$) have similar bell-shaped form centered in a neighborhood of the expectation, while the plot for $\delta=2.0$ has an additional rather sharp local maximum at the origin. A similar behavior of the probability densities of the entanglement entropy is presented on Fig. \ref{fig3} for the exponential (\ref{ex}) and the "half"-Cauchy (\ref{ca}) on-site probability distributions of the potential. The bell- shaped forms could be a manifestation of finite size effects which seem likely in the case of the uniform distribution of Fig. 2a) and the weak disorder $\delta=0.2$ and, possibly,
$\delta=0.4$ where the localization radius at the Fermi energy is comparable with the block length according to Fig. 2b) and Table 1.  However, we have the bell-shaped  plots in Fig. 2a) for $\delta=0.7, \ 1.0$,  Fig. 3a) for $\delta=0.4, \ 0.7$ and Fig. 3b) for $\delta=0.2$ for which the localization radius is one or even two order of magnitude less than the length of block and the length of the whole system. This suggests a certain amount of universality of the bell-shaped form of the probability density of the entanglement entropy for the weak and medium disorder. Indeed, we found that all these plots fit sufficiently well the Gaussian probability density centered at the expectation $\mathbf{E}\{S_{[-M,M]}\}$ with the precision within 1\% - 4\%. Thus, one can guess a kind of intermediate asymptotic form of the probability
density given by a certain Central Limit Theorem.

\begin{figure}[ht]
%\label{delta}
\begin{minipage}[ht]{0.5\linewidth}
\center{\includegraphics[width=0.99\linewidth]{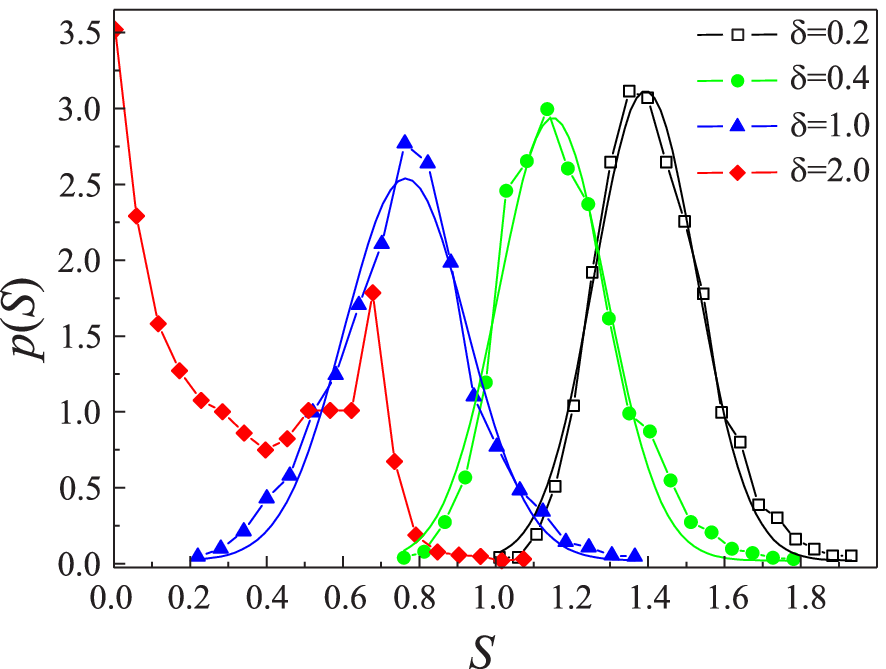} \\ a)}
\end{minipage}
\hfill
\begin{minipage}[ht]{0.5\linewidth}
\center{\includegraphics[width=0.99\linewidth]{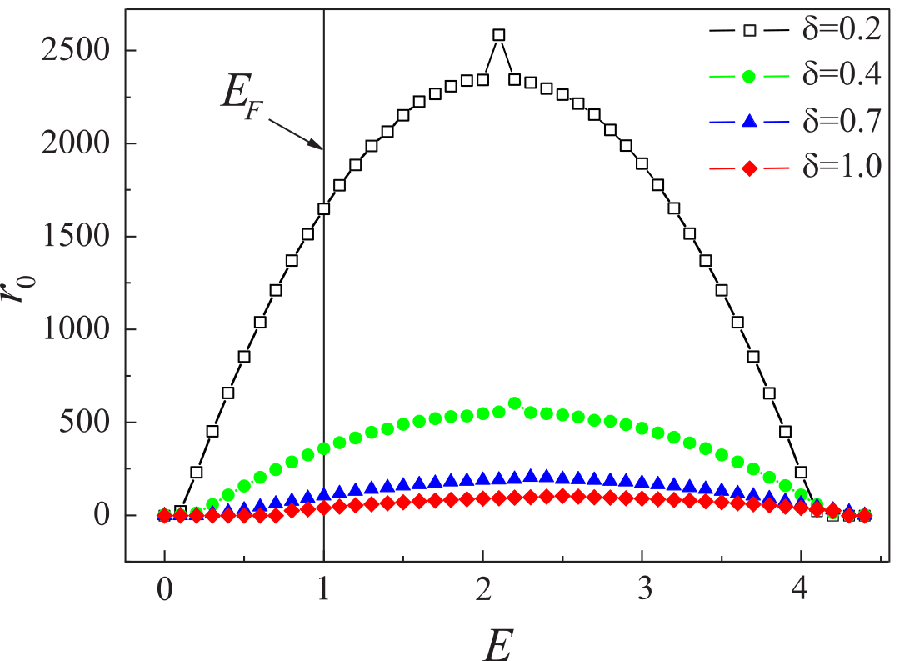} \\ b)}
\end{minipage}
\caption{a) The probability density of the
entanglement entropy $S_{[-M,M]}$ for different values
of disorder parameter $\delta$ and the uniform on-site distribution
(\ref{un}). b) The localization radius as a function of the
Fermi energy $E$. The  vertical line indicates the value $E=1$ of
the Fermi energy used in Fig. \ref{fig11}, Fig. \ref{fig2}a) and
Fig. \ref{fig3} -- Fig. \ref{fig5}. Solid lines without symbols are the Gaussian fittings.}
\label{fig2}
\end{figure}

Now about the maximum near the origin of the probability density of the entanglement entropy in Fig. 2a) for $\delta=2.0$, Fig. 3a) for $\delta=1.0$ and Fig. 3b) for $\delta=0.7, \ 1.0$. This is one more manifestation of the decay of the entanglement entropy as disorder grows. It follows from the results of \cite{Ai-Wa:15,El-Co:17,Pa-Sl:14} that
\begin{equation}\label{edis}
\mathbf{E}\{S_{[-M,M]}\} \le C /\delta^a,
\end{equation}
where $C < \infty$ and $a>0$ do not depend on $M$ and $\delta$. Since $S_{[-M,M]}$ is non-negative, the above bound and the Chebyshev inequality imply an analogous bound (with $a' < a$) for the probability of non-zero values of $S_{[-M,M]}$.

Fig. \ref{fig4} shows the probability densities of the R\'enyi entropy (\ref{Rg}) -- (\ref{hal}) for $\alpha \ge 1$.  Here the plots, being qualitatively similar to those of Fig. 2a and Fig. 3 for the von Neumann entropy with the same values of the disorder parameter, are closer to the origin, since the R\'enyi entropy decreases monotonically as $\alpha \ge 1$ increases.

Fig. \ref{fig5} is the graphic description of the finite size version of the bound (\ref{vart}), written as
\begin{equation}\label{cvs}
\mathbf{C_V}\{S_{[-M,M]}\} \ge \mathbf{C_V}(t)
\end{equation}
where  $\mathbf{C_V}\{S_{[-M,M]}\}$ is the coefficient of variation (\ref{cv}) of the entanglement entropy and
\begin{equation}\label{cvt}
\mathbf{C_V}(t)=\sqrt{2}\left|1-\mathbf{E}\{S_{[-M,M]}^{t}\}/\mathbf{E}\{S_{[-M,M]}\}\right|/F^{1/2}(t)
\end{equation}
is the square root of the r.h.s of the finite size version of (\ref{vart}) divided by $\mathbf{E}\{S_{[-M,M]}\}$.

\begin{figure}[ht]
%\label{alpha}
\begin{minipage}[ht]{0.5\linewidth}
\center{\includegraphics[width=0.99\linewidth]{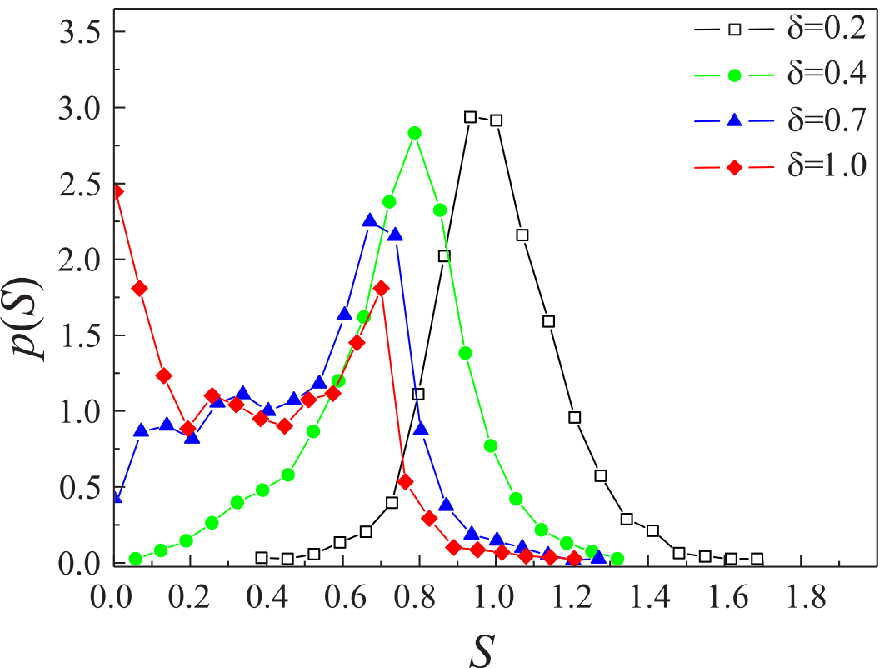} \\a)}
\end{minipage}
\hfill
\begin{minipage}[ht]{0.5\linewidth}
\center{\includegraphics[width=0.99\linewidth]{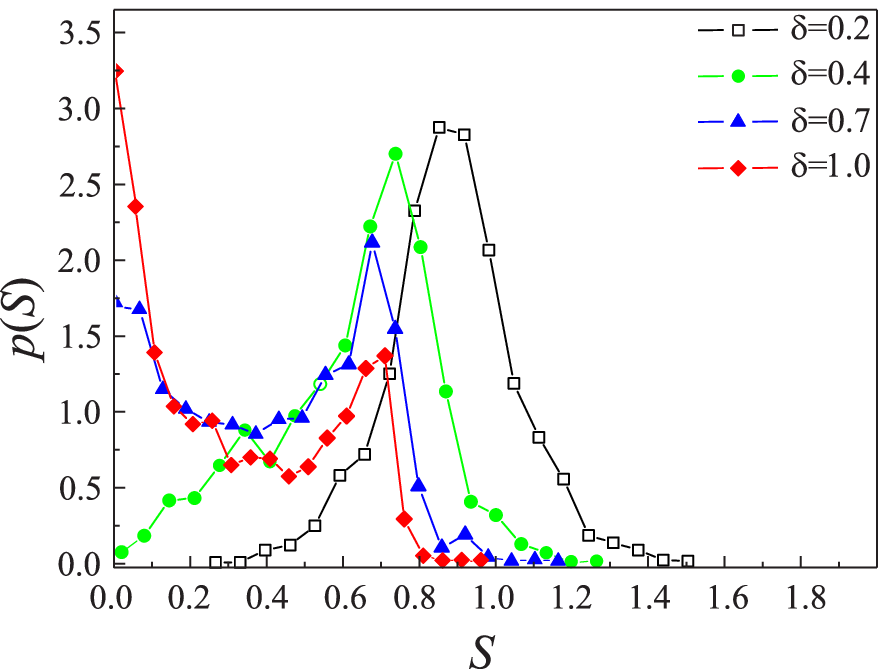} \\ b)}
\end{minipage}
\caption{The probability density of the entanglement entropy $S_{[-M,M]}$
with different disorder parameters $\protect\delta$ for: a) exponential
distribution (\ref{ex}); b) "half"-Cauchy distribution (\ref{ca})}
\label{fig3}
\end{figure}

It follows from the figure that our lower bounds (\ref{var}) -- (\ref{F}) and (\ref{vart}) are sufficiently exact although not tight (the maxima of the both plots are about 92 percents of the of the corresponding coefficient of variation, indicated by the corresponding horizontal lines).

\begin{figure}[ht]
\begin{minipage}[ht]{0.5\linewidth}
\center{\includegraphics[width=0.99\linewidth]{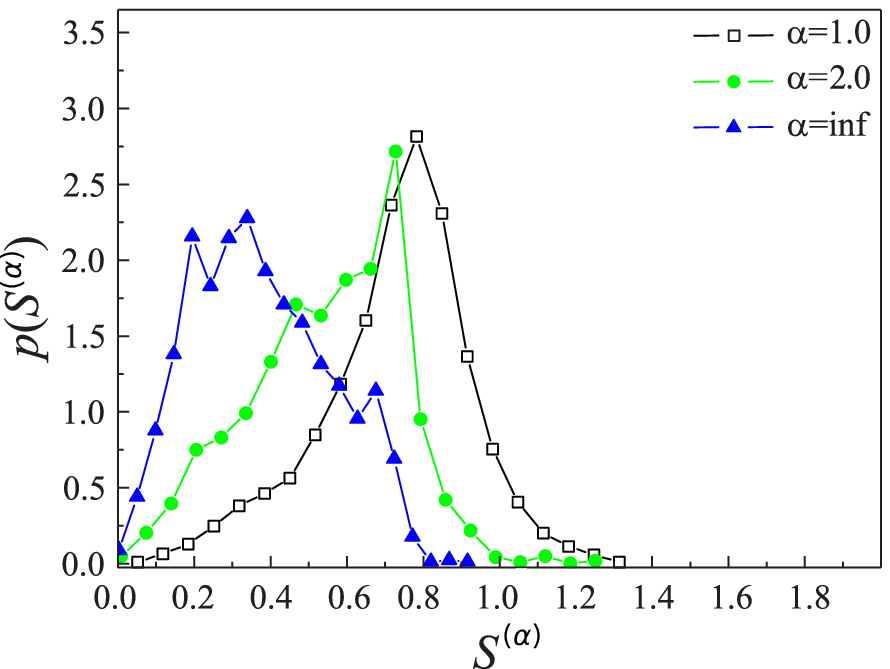} \\a)}
\end{minipage}
\hfill
\begin{minipage}[ht]{0.5\linewidth}
\center{\includegraphics[width=0.99\linewidth]{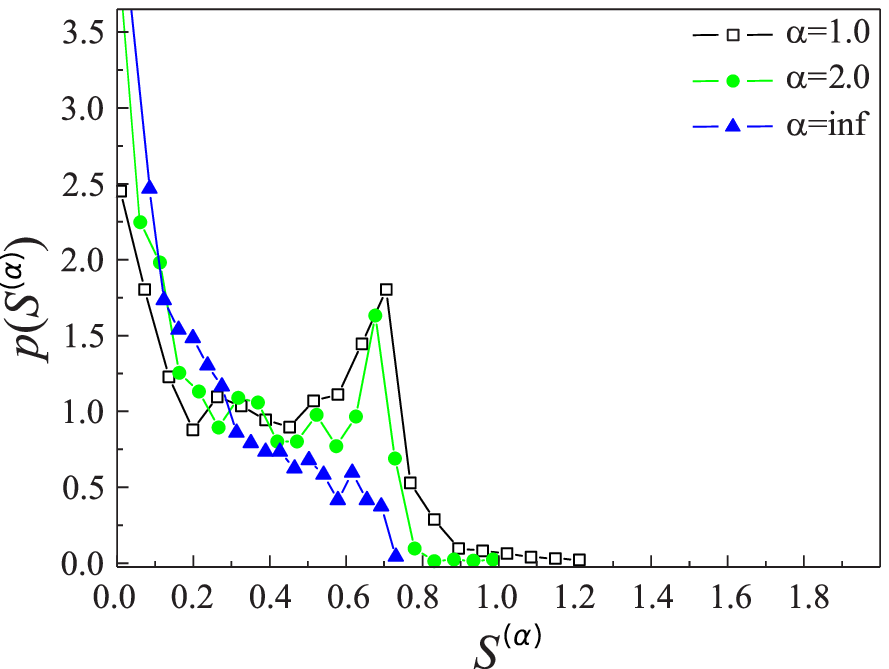} \\ b)}
\end{minipage}
\caption{The probability densities  of the R\'enyi entropy $ S^{(\alpha)}_{[-M,M]}$ of (\ref{Rf}) -- (\ref{hal}) for the exponential on-site distribution (\ref{ex}) with
different values of  parameter $\alpha$ for:
a)  $\delta=0.4$; b)  $\delta=1.0$.}
\label{fig4}
\end{figure}

\section{Auxiliary results}

\begin{lemma}
\label{l:cr} Let $\xi $ be a non-negative random variable, $\varphi :\mathbb{%
R}_{+}\rightarrow \mathbb{R}$ be a function and $\mathbf{E}\{|\varphi (\xi
)|\}<\infty $. Assume that the probability law of $\xi $ has a bounded
density $f$ such that

(i) $\mathrm{supp}\,f=[0,\infty );$

(ii) the quantity $F(t)$ in (\ref{F}) %\begin{equation*}
%F(t)=I(t)-1,\;I(t)=\int_{0}^{\infty }\frac{f^{2}(v-t)}{f(v)}%
%dv<\infty
%\end{equation*}
is well defined for some $t>0$.

\noindent Then we have%
\begin{equation}
\mathbf{Var}\{\varphi (\xi )\}\geq \left. (\mathbf{E}\{\varphi (\xi
)-\varphi (\xi+ t )\})^2\right/ F(t).  \label{cr}
\end{equation}
\end{lemma}

\begin{figure}[ht]
\center{\includegraphics[width=8.0cm]{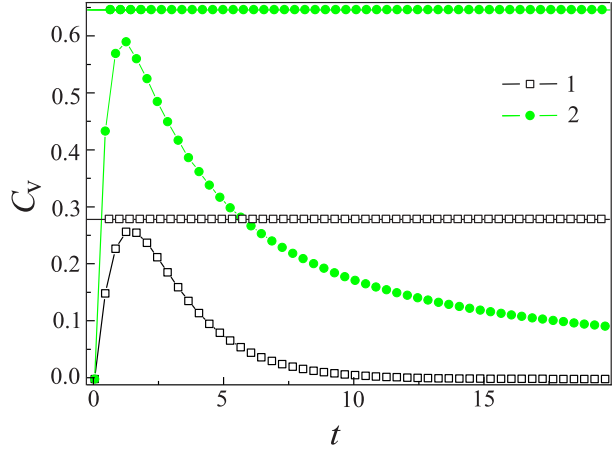}}
\caption{The r.h.s.   $\mathbf{C_{V}}(t)$ of (\ref{cvs}) (see also (\ref{cvt})).
Curve 1 corresponds to exponential on-site distribution (\ref{un}) with $\delta=0.4$
and
curve 2 corresponds to half-Cauchy on-site distribution (\ref{ca}) with $\delta=0.7$.
Horizontal lines indicate the corresponding values of the l.h.s.
$\mathbf{C_{V}}\{S_{[-M,M]}\}$ of (\ref{cv}), the coefficient of variation
of the entanglement entropy.}
\label{fig5}
\end{figure}

\begin{proof}
Consider the random variables $\varphi (\xi )$ and $\psi (\xi )=(f(\xi
-t)-f(\xi ))/f(\xi )$. It follows from the normalization condition%
\begin{equation}  \label{nor}
\int_{0}^{\infty }f(x)dx=1
\end{equation}%
that
\begin{equation}\label{epsi}
\mathbf{E}\{\psi (\xi )\}=\int_{0}^{\infty }\left( \frac{f(x-t)-f(x)}{f(x)}%
\right) f(x)dx=0.
\end{equation}%
Thus,%
\begin{eqnarray*}
\mathbf{Cov}\{\varphi (\xi )\psi (\xi )\} :=& \mathbf{E}\{(\varphi (\xi )-
\mathbf{E}\{\varphi (\xi )\})(\psi (\xi )-\mathbf{E}\{\varphi (\xi )\})\} \\
&\hspace{2cm}=\mathbf{E}\{\varphi (\xi )\psi (\xi )\}.
\end{eqnarray*}%
On the other hand, we have from the Schwarz inequality for the expectations:%
\begin{align*}
(\mathbf{Cov}\{\varphi (\xi )\psi (\xi )\})^{2}\leq \mathbf{Var}\{\varphi
(\xi )\}\mathbf{Var}\{\psi (\xi )\}.
\end{align*}%
Combining these two relations and using the definition of $\psi $ implying (\ref{epsi}) and%
\begin{eqnarray}
\mathbf{Var}\{\psi (\xi )\} &=&\int_{0}^{\infty }\left( \frac{f(x-t)-f(x)}{%
f(v)}\right) ^{2}f(x)dx  \label{Fpos} \\
&=&\int_{0}^{\infty }\frac{f^{2}(x-t)}{f(x)}dx-1=F(t)\geq 0,  \notag
\end{eqnarray}%
we get (\ref{cr}).
\end{proof}

\begin{remark}
\label{r:cr} The inequality is, in fact, the Hammersley-Chapman-Robbins
inequality (see \cite{Le-Ca:98}, Section 2.5.1) and is a version of the
Cram\'er-Rao inequality of statistics.
\end{remark}

\begin{lemma}
\label{l:pvan} Let $H$ be the one-dimensional Schr\"{o}dinger operator (see (%
\ref{SD}) -- (\ref{pot}) for $d=1$) with an i.i.d. non-negative random
potential $V=\{V(x)\}_{x\in \mathbb{Z}}$ whose common probability law has a
density $f$ satisfying (\ref{cl2}). Denote $P^{t}=\mathcal{E}%
_{H^{t}}((0,E))=\{P^{t}(x,y)\}_{x,y\in \mathbb{Z}}$ the Fermi projection of $%
H^{t}$ of (\ref{ht}) corresponding to the spectral interval $(0,E)$.

Then there exist $s\in (0,1),\;C<\infty $ and $\gamma >0$ that do not depend
on $t > 0$ and are such that we have for $0<E<t$:

\medskip (i) \ $\mathbf{E}\{|P^{t}(0,0)|\}\leq C(t-E)^{-s}$;

\medskip (ii) $\mathbf{E}\{|P^{t}(x,y)|\}\leq C(t-E)^{-s}e^{-\gamma (x-y)},$
for all $x\geq 1$ and $y\leq -1$.
\end{lemma}

\begin{proof}
Let%
\begin{equation*}
G^{t}(z)=(H^{t}-z)^{-1}=\{G^{t}(x,y;z)\}_{x,y\in \mathbb{Z}},\;z\in \mathbb{C%
}\setminus \mathbb{R}
\end{equation*}%
be the resolvent of $H^{t}$. It is shown below that the bounds%
\begin{equation}
\mathbf{E}\{|G^{t}(0,0;\lambda +i\eta )|^{s}\}\leq C/(t-E)^{s}  \label{eg00}
\end{equation}%
and%
\begin{equation}
\mathbf{E}\{|G^{t}(x,y;\lambda +i\eta )|^{s}\}\leq C/(t-E)^{s}e^{-\gamma
(x-y)}  \label{egxy}
\end{equation}%
are valid for some $s\in (0,1)$, all $\lambda \in (0,E),\;0 < E<t,\;\eta
\neq 0$ and $x\geq 1$, $y\leq -1$ with $C<\infty $ and $\gamma >0$ which are
independent of $t,\;\lambda $ and $\eta \neq 0$.

It follows from a slightly modified version of proof of Theorem 13.6 of \cite%
{Ai-Wa:15}, based on the contour integral representation of $P^t$ via $G^t$
and the Combes-Thomas theorem, that the assertion of the lemma can be
deduced from (\ref{eg00}) -- (\ref{egxy}).

Hence, it suffices to prove (\ref{eg00}) and (\ref{egxy}). To this end we
introduce the restrictions $H_{-}$ and $H_{+}$ of $H^{t}$ (or $H$) to the
integer-valued intervals $(-\infty ,-1]$ and $[1,\infty )$ and the rank one
operator $\widehat{V}_{0}^{t}$ of multiplication by $V(0)+2+t$. Let%
\begin{equation}
\widehat{H}=H_{-}\oplus \widehat{V}_{0}^{t}\oplus H_{+}  \label{hh}
\end{equation}%
be the double infinite block matrix consisting of the $(-\infty ,-1]\times
(-\infty ,-1]$ semi-infinite block $H_{-}$, $1\times 1$ "central" block $%
\widehat{V}_{0}^{t}$ and the $[1,\infty )\times \lbrack 1,\infty )$
semi-infinite block $H_{+}$. In other words, $\widehat{H}^{t}$ is obtained
from $H^{t}$ by replacing its four entries (equal $-1)$ with indices $(0,\pm
1)$ and $(\pm 1,0)$ by zero. Denote
\begin{align}
\widehat{G}^{t}(z)& =(\widehat{H}^{t}-z)^{-1}=\{\widehat{G}%
^{t}(x,y)\}_{x,y\in \mathbb{Z}},  \label{gpm} \\
G_{\pm }(z)& =(H_{\pm }-z)^{-1}=\{G_{\pm }(x,y)\}_{x,y\in \lbrack \pm 1,\pm
\infty )}  \notag
\end{align}%
the corresponding resolvents, where we omit the complex spectral parameters $%
z,\;\Im z\neq 0$ in the r.h.s. We have in view of (\ref{hh}):%
\begin{equation}
\widehat{G}^{t}=G_{-}\oplus (\widehat{V}_{0}^{t}-z)^{-1}\oplus G_{+}.
\label{ght1}
\end{equation}%
By using the resolvent identity $G^{t}=\widehat{G}^{t}-\widehat{G}^{t}(H^{t}-%
\widehat{H}^{t})G^{t}$, we obtain for all $x,y\in \mathbb{Z}$%
\begin{eqnarray*}
G^{t}(x,y) &=&\widehat{G}^{t}(x,y)+\widehat{G}%
^{t}(x,0)(G^{t}(-1,y)+G^{t}(1,y)) \\
&&+(\widehat{G}^{t}(x,-1)+\widehat{G}^{t}(x,1))G^{t}(0,y).
\end{eqnarray*}%
This and (\ref{ght1}) imply%
\begin{equation}
G^{t}(0,y)=(V(0)+2+t-z)^{-1}(\delta _{0,y}+G^{t}(-1,y)+G^{t}(1,y)),\;y\leq 0
\label{gt0y}
\end{equation}%
and
\begin{equation}
G^{t}(x,y)=G^{t}(0,y)G_{+}(x,1),\;x\geq 1,\;y\leq -1.  \label{gt1}
\end{equation}%
Likewise, we have from the resolvent identity $G^{t}=\widehat{G}%
^{t}-G^{t}(H^{t}-\widehat{H}^{t})\widehat{G}^{t}$:%
\begin{equation}
G^{t}(x,0)=(V(0)+2+t-z)^{-1}(\delta _{x,0}+G^{t}(x,-1)+G^{t}(x,1)),\;x\geq 0
\label{gtx0}
\end{equation}%
and%
\begin{equation}
G^{t}(x,y)=G^{t}(x,0)G_{-}(-1,y),\;x\geq 1,y\leq -1.  \label{gt2}
\end{equation}%
We have then from (\ref{gt0y}) and the Schwarz inequality for any $s>0$ and
all $y\leq 0$
\begin{eqnarray}
\mathbf{E}\{|G^{t}(0,y)|^{s}\} &\leq &C_{s}\,(g(z-2-t))^{1/2}  \label{egt0y}
\\
&&(\delta _{0,y}+\mathbf{E}^{1/2}\{|G^{t}(-1,y)|^{2s}\}+\mathbf{E}%
^{1/2}\{|G^{t}(1,y)|^{2s}\}),  \notag
\end{eqnarray}%
where $C_{s}$ depends only on $s>0$ and
\begin{equation*}
g(\zeta )=\mathbf{E}\{|V(0)-\zeta |^{-2s}\}=\int_0^\infty \frac{f(v)dv}{%
|v-\zeta |^{2s}}.
\end{equation*}%
Choosing here $\zeta =z-2-t,\;z=\lambda +i\eta $ and using (\ref{nor}), we
have for $\lambda \leq E<t$
\begin{equation}
g(z-2-t)\leq (t-E)^{-2s}  \label{gb}
\end{equation}%
and then (\ref{egt0y}) implies for any $s>0$ and all $y\leq 0$
\begin{align}
& \mathbf{E}\{|G^{t}(0,y;z)|^{s}\}\leq C_{s}(t-E)^{-s}  \label{gt0ys} \\
& \hspace{0.5cm}\times (\delta _{0,y}+\mathbf{E}^{1/2}\{|G^{t}(-1,y;z)|^{2s}%
\}+\mathbf{E}^{1/2}\{|G^{t}(1,y;z)|^{2s}\}).  \notag
\end{align}%
We will use now Theorem 8.7 of \cite{Ai-Wa:15}, according to which if $A_{0}$
is a selfadjoint operator in $l^{2}(\mathbb{Z}^{d})$, $U=\{U(x)\}_{x\in
\mathbb{Z}^{d}}$ is a collection of independent random variables whose
probability densities $f_x,\, x \in \mathbb{Z}$ are bounded uniformly in $x$%
, i.e., $\sup_{x \in \mathbb{Z}} \sup_{v \in \mathbb{R}} < \infty$ and if $%
\mathcal{G}(z)=(A-z)^{-1}=\{\mathcal{G} (x,y;z)\}_{x,y\in \mathbb{Z}^{d}}$
is the resolvent of $A=A_{0}+U$, then for any $s\in (0,1)$ there exists $%
C_{s}^{\prime }<\infty $ such that the bound%
\begin{equation*}
\mathbf{E}\{|\mathcal{G}(x,y;z)|^{s}\}\leq C_{s}^{\prime }
\end{equation*}%
holds uniformly in $z\in \mathbb{C\setminus R}$\ for all $x,y\in \mathbb{Z}%
^{d}$.

Choosing here $A=H^{t}$ with $H^{t}$ of (\pageref{ht}) and noting that for
the potential $V^t$ the conditions of the theorem are satisfied (all $%
V^t(x)=V(x), \, x\neq 0$ are i.i.d random variables with a bounded common
probability density and $V^t(0)=V(0)+t$ has the density $f(v-t)$ also
bounded), we obtain for any $s\in (0,1),\;\lambda <E, \; \eta \neq 0$ and
all $x,y\in \mathbb{Z}^{d}$
\begin{equation}
\mathbf{E}\{|G^{t}(x,y;\lambda +i\eta )|^{s}\}\leq C_{s}^{\prime \prime },
\label{gtxyb}
\end{equation}%
where $C_{s}^{\prime \prime }$ does not depend on $t,E$ and $\eta $.

Plugging this bound with $x=\pm 1$ into (\ref{gt0ys}), we get for any $%
s_{1}\in (0,1/2),\;\lambda <E<t$ and all $y\leq 0$
\begin{equation}
\mathbf{E}\{|G^{t}(0,y;\lambda +i\eta )|^{s_{1}}\}\leq
B_{s_{1}}/(t-E)^{s_{1}},  \label{eg0yf}
\end{equation}%
where $B_{s_{1}}$ does not depend on $t,E$ and $\eta $.

Analogous argument yields for $s_{1}\in (0,1/2),\;\lambda <E<t$ and all $%
x\geq 0$
\begin{equation}
\mathbf{E}\{|G^{t}(x,0;\lambda +i\eta )|^{s_{1}}\}\leq
B_{s_{1}}/(t-E)^{s_{1}}.  \label{egx0f}
\end{equation}%
We obtain (\ref{eg00}), hence assertion (i) of the lemma, from (\ref{eg0yf})
with $y=0$ (or from (\ref{egx0f}) with $x=0$).

To prove (\ref{egxy}), hence assertion (ii) of the lemma, we combine (\ref%
{gt1}) and (\ref{gt2}) to write for any $s>0$ and all $x\geq 1$ and $y\leq 1$%
:
\begin{align*}
|G^{t}(x,y;z)|^{s}& =|G^{t}(x,0;z)|^{s/2}|G^{t}(0,y;z)|^{s/2} \\
\times| & G_{+}(x,1;z)|^{s/2}|G_{-}(-1,y;z)|^{s/2}
\end{align*}%
and then, by H\"{o}lder inequality for expectations,%
\begin{eqnarray}
&&\mathbf{E}\{|G^{t}(x,y;z)|^{s}\}\leq \mathbf{E}^{1/4}\{|G^{t}(x,0;z)|^{2s}%
\}\mathbf{E}^{1/4}\{|G^{t}(0,y)|^{3s}\}  \label{g4} \\
&&\hspace{5cm}\times \mathbf{E}^{1/4}\{|G_{+}(x,1)|^{2s}\}\mathbf{E}%
^{1/4}\{|G_{-}(-1,y)|^{2s}\}.  \notag
\end{eqnarray}%
Using here (\ref{eg0yf}) and (\ref{egx0f}), we get for $s_{1}\in (0,1/2)$,
and all $x\geq 1$ and$\;y\leq -1$

\begin{equation}
\mathbf{E}\{|G^{t}(x,y;z)|^{s_{1}}\}\leq \frac{B_{s_{1}}}{(t-E)^{s_{1}}}%
\mathbf{E}^{1/4}\{|G_{+}(x,1)|^{2s_{1}}\}\mathbf{E}^{1/4}%
\{|G_{-}(-1,y)|^{2s_{1}}\}.  \label{gsi}
\end{equation}%
To bound the two last factors on the right, we will use a result from \cite%
{Mi:96} according to which if $H_{+}$ is the discrete one dimensional Schr%
\"{o}dinger operator in $l^{2}([1,\infty ))$ with i.i.d. potential whose
common probability law is such that $\mathbf{E}\{|V(1)|^{\kappa }\}<\infty $
for some $\kappa >0$, then %\begin{equation*}
%\int_{-\infty} ^{\infty} |v|^{k}f(v)dv<\infty ,
%\end{equation*}%
for any spectral interval $I$ there exist $C(I)<\infty ,\,\ \gamma (I)>0 $
and $s_{2}\leq \kappa /2$ such that%
\begin{equation}
\mathbf{E}\{|G_{+}(x,1)|^{s_{_{2}}}\}\leq C(I)e^{-\gamma (I)x},\;x\geq 1.
\label{min1}
\end{equation}%
The same is valid for the Hamiltonian $H_{-}$ acting in $l^{2}((-\infty
,-1]) $, see (\ref{hh}):
\begin{equation}
\mathbf{E}\{|G_{+}(-1,y)|^{s_{_{2}}}\}\leq C(I)e^{\gamma (I)y},\;y\leq -1.
\label{min2}
\end{equation}%
The bounds are the basic ingredient of the proof of (\ref{b_P}) for the one
dimensional case \cite{Mi:96}.

Using these bounds in the r.h.s. of (\ref{gsi}), we obtain assertion (ii) of
the lemma with $s=\min \{s_{1},s_{2}\}$, where $s_{1}$ and $s_{2}$ are
defined in (\ref{eg0yf}) -- (\ref{egx0f}) and (\ref{min1}) -- (\ref{min2}).
\end{proof}

\begin{remark}
\label{r:weyl} (i) By using the standard facts of spectral theory, it is
easy to prove the weaker version of Lemma \ref{l:pvan}%
\begin{equation}
\lim_{t\rightarrow \infty }P^{t}(x,y)=0  \label{lp}
\end{equation}%
valid for every $x\geq 0,\,y\leq 0$. This, however, does not allow us to justifies the limiting transition $%
t\rightarrow \infty $ in the second term in the r.h.s. of (\ref{espt}).

Indeed, according to the spectral theorem%
\begin{equation*}
G^{t}(x,y;z)=\int_{-\infty }^{\infty }\frac{\mathcal{E}_{H^{t}}(x,y;d\lambda
)}{\lambda - z},\,\, \Im z\neq 0.
\end{equation*}%
The formula, the continuity properties of the Stieltjes transform of a
bounded signed measure and (\ref{fp}) imply that (\ref{lp}) follows from the
analogous limiting relation for the resolvent with $\Im z\neq 0$:%
\begin{equation}
\lim_{t\rightarrow \infty }G^{t}(x,y;z)=0,\,\, x\geq 0,\,y\leq 0.  \label{lg}
\end{equation}%
Viewing the term $t$ of the $(0,0)$th entry $V(0)+t$ of $H^{t}$ in (\ref{ht}%
) as a rank one perturbation of $H=H^{t}|_{t=0}$, we obtain%
\begin{equation*}
G^{t}(x,y;z)=G(x,y;z)-\frac{tG(x,0;z)G(0,y;z)}{1+tG(0,0;z)},
\end{equation*}%
hence%
\begin{equation}
\lim_{t\rightarrow \infty }G^{t}(x,y;z)=G(x,y;z)-\frac{G(x,0;z)G(0,y;z)}{%
G(0,0;z)}.  \label{lgt}
\end{equation}%
Recall now the Weyl formula for the resolvent of the discrete Schr\"odinger
operator (see, e.g. \cite{Te:99}, Section 1.2):%
\begin{equation*}
G(x,y;z)=G(0,0;z)\left\{
\begin{array}{cc}
\psi _{+}(x;z)\psi _{-}(y;z), & x\geq y, \\
\psi _{-}(x;z)\psi _{+}(y;z), & x\leq y,%
\end{array}%
\right.
\end{equation*}%
where $\psi _{\pm}$ are the solutions of the corresponding discrete
Schr\"odinger equation which belong to $l^{2}(\mathbb{Z}_{\pm })$ for $\Im z
\neq 0 $ and satisfy the condition $\psi_{\pm }(0)=1$. Combining the formula
with (\ref{lgt}), we obtain (\ref{lg}), hence (\ref{lp}).

\medskip (ii) The scheme of proof of an analog of (\ref{var}) -- (\ref{AF})
for the R\'enyi entanglement entropy $S_{[-M,M]}^{(\alpha)}$ of (\ref{Rf})
-- (\ref{hal}) is as follows. Analogs of (\ref{sp1}) -- (\ref{spm}) for $%
S_{[-M,M]}^{(\alpha)}$ with $\alpha \in (0,1]$ are proved in \cite{El-Co:17}.
The proof is rather involved since for $\alpha<1$ the function $h_{\alpha}$
of (\ref{hal}) is H\"older continuous with the H\"older exponent $\alpha <1$%
. The case $\alpha \ge 1$ where $h_{\alpha}$ is H\"older continuous with the
H\"older exponent 1 is just a particular and rather simple case of the
general case of \cite{El-Co:17}. We then start from this result and repeat
almost literally the above proof of (\ref{var}) -- (\ref{AF}) dealing with
the von Neumann entropy.
\end{remark}

%\bigskip \noindent
\textbf{Acknowledgment} The work is supported in part by the grant 4/16-M of
the National Academy of Sciences of Ukraine.

%%---------------------------------------------------------------

\end{document}